\documentclass[a4paper,11pt]{article}
\usepackage{amsfonts}
\usepackage{cite}
\usepackage{graphics,amsmath,amssymb,amsthm,accents}

\def \h#1{\widehat{#1}}
\def \t#1{\widetilde{#1}}
\def \b#1{\overline{#1}}

\def \th#1{\widehat{\widetilde{#1}}}
\def \hb#1{{\widehat{\overline{#1}}}}

\def \tb#1{\widetilde{\overline{#1}}}

\numberwithin{equation}{section}
\newtheorem{prop}{Proposition}
\newcommand{\balpha}{\boldsymbol{\alpha}}
\newcommand{\bfg}{\mathbf{g}}
\newcommand{\bfy}{\mathbf{y}}
\newcommand{\bft}{\mathbf{t}}

\title{Deriving conservation laws for ABS lattice equations\\ from Lax pairs }

\author{Da-jun Zhang\footnote{Corresponding author. E-mail: djzhang@staff.shu.edu.cn},
~~ Jun-wei Cheng,~~Ying-ying Sun\\
{\small\it Department of Mathematics, Shanghai University, Shanghai 200444,  P.R. China}}

\begin{document}

\maketitle

\begin{abstract}

In the paper we derive infinitely many conservation laws for the ABS
lattice equations from their Lax pairs. These conservation laws can
algebraically be expressed by means of some known polynomials. We
also show that H1, H2, H3, Q1, Q2, Q3 and A1 equation in ABS list
share a generic discrete  Riccati equation.

\vskip 6pt
\noindent
{\bf Keywords}: conservation laws, Lax pairs, ABS lattice equations, discrete integrable systems

\vskip 6pt
\noindent
{PACS:} 02.30.Ik, 02.90.+p\\
{MSC:} 39-04, 39A05, 39A14
\end{abstract}

\section{Introduction}

The celebrated ABS list\cite{ABS-2003} consists of totally  nine
quadrilateral lattice equations which are consistent around a
cube(CAC) (3D-consistency with extra restriction: quasi-linearity,
$D_4$ symmetry and Tetrahedron property).
A CAC equation itself provides a B\"{a}cklund transformation(BT) as
well as a Lax pair\cite{ABS-2003}. However, such a Lax pair seems
hard to play a same role as useful as in continuous case, where from
Lax pairs usually one can derive solutions either through Inverse
Scattering Transform\cite{Ab-IST} or Darboux
transformation\cite{Matveev-DT}, construct BTs\cite{wsk}, evolution
equation hierarchy, commutative flows and infinitely many
symmetries\cite{LZ-1986} and infinitely many conservation
laws\cite{wsk}.

With regard to infinitely many conservation laws which serve as an
integrability characteristic, there are many ways to derive them for
continuous and semi-discrete integrable
systems\cite{wsk,ko-cl,za-sh,ts-wa,zc,zwxz}. Using Lax pair is a
simple way\cite{wsk,zc,zwxz}. For  ABS lattice equations, their
conservation laws have been derived through a direct approach\cite{rh}
based on the idea of \cite{h}, symmetry
approach\cite{RH-JPA-2007,x-2011,rs-2009}, Gardner's approach (using
BTs and initial conservation laws)\cite{rs-2009,Gardner-1}, and
using quasi-difference operators and recursion
operators\cite{mwx-2011}, etc. In this paper, we will start from Lax
pairs to derive infinitely many conservation laws for ABS lattice
equations. In Ref.\cite{zc} (also see \cite{zwxz}) we introduce two
kinds of techniques to construct conservation laws respectively for
the Toda lattice and Ablowitz-Ladik system. In fact,  Gardner method
used in \cite{Gardner-1} is closely related to the technique used
for the Ablowitz-Ladik system\cite{zc}. However, one will see that
we can easily write out the so-called initial conservation laws from
Lax pairs, and infinitely many conservations laws can algebraically
be expressed by means of known polynomials. Besides, we also find
many ABS lattice equations share a generic discrete  Riccati
equation.

We organize the paper as follows. In Sec.2 we introduce ABS list and
the main idea of our approach. In Sec.3 first H1 equation serves as
a detailed example to derive conservation laws. Then we list main
results for H2, H3, Q1, Q2, Q3 and A1 equation. Finally, also in
this section, we derive conservation laws for A2 and Q4 equation in
a slightly different way but still starting from Lax pairs.

\section{Preliminary and main idea}\label{sec:2}

Let us start from the following quadrilateral equation
\begin{equation}\label{deq}
Q(u,\t{u},\h{u},\th{u},p,q)=0,
\end{equation}
where
\[u=u(n,m),~~\t{u}=E_nu=u(n+1,m),~~\h{u}=E_mu=u(n,m+1),~~\th u=u(n+1,m+1),\]
$E_n$ and $E_m$ respectively serve as  shift operators in direction
$n$ and $m$, $p$ and $q$ are spacing parameters of direction $n$ and
$m,$ respectively.

The ABS list reads\cite{ABS-2003}
\begin{align*}
& (u-\th{u})(\t{u}-\h{u})+q-p=0, \tag{H1}\\
& (u-\th{u})(\t{u}-\h{u})+(q-p)(u+\t{u}+\h{u}+\th{u})+q^2-p^2=0, \tag{H2}\\
& p(u\t{u}+\h{u}\th{u})-q(u\h{u}+\t{u}\th{u})+\delta(p^2-q^2)=0, \tag{H3}\\
& p(u+\h{u})(\t{u}+\th{u})-q(u+\t{u})(\h{u}+\th{u})-\delta^2pq(p-q)=0, \tag{A1}\\
& (q^2-p^2)(u\t{u}\h{u}\th{u}+1)+q(p^2-1)(u\h{u}+\t{u}\th{u})-p(q^2-1)(u\t{u}+\h{u}\th{u})=0, \tag{A2}\\
& p(u-\h{u})(\t{u}-\th{u})-q(u-\t{u})(\h{u}-\th{u})+\delta^2pq(p-q)=0, \tag{Q1}\\
& p(u-\h{u})(\t{u}-\th{u})-q(u-\t{u})(\h{u}-\th{u})+pq(p-q)(u+\t{u}+\h{u}+\th{u})\notag\\
& ~~~~~ -pq(p-q)(p^2-pq+q^2)=0, \tag{Q2}\\
& (q^2-p^2)(u\th{u}+\t{u}\h{u})+q(p^2-1)(u\t{u}+\h{u}\th{u})-p(q^2-1)(u\h{u}+\t{u}\th{u})\notag\\
& ~~~~~ -\delta^2(p^2-q^2)(p^2-1)(q^2-1)/(4pq)=0, \tag{Q3}\\
& p(u\t{u}+\h u\th{u})-q(u\h u+\t u\th u)=\frac{pQ-qP}{1-p^2q^2}\bigl[(\h u\t u+u\th u)-pq(1+u\t u\h u\th u)\bigr], \tag{Q4}
\end{align*}
where Q4 equation is of the form given by
Hietarinta\cite{Hietarinta-CAC}, and in the equation $P^2=p^4-k p^2+1,~Q^2=q^4-k q^2+1$
with parameter $\delta$.

If \eqref{deq} is a CAC equation, then it is easy to write out its
BT
\begin{subequations}\label{BT}
\begin{align}
& Q(u,\t{u},\b{u},\tb{u},p,r)=0,\\
& Q(u,\b{u},\h{u},\hb{u},r,q)=0,
\end{align}
\end{subequations}
where $r$ serves as a soliton parameter, and if $u$ solves
\eqref{deq}, so does $\b u$. Replacing $\b u$ by $\phi_1/\phi_2$,
the above BT can be rewritten in terms of $\phi=(\phi_1,\phi_2)^T$
as the following,
\begin{subequations}
\label{Lax}
\begin{align}
& \t{\phi}=\frac{\beta}{\sqrt{|M|}}M(u,\t u, p,r)\phi,\label{Lax-a}\\
& \h \phi =\frac{\gamma}{\sqrt{|N|}}N(u,\h{u},q,r)\phi,\label{Lax-b}
\end{align}
\end{subequations}
where $M$ and $N$ are $2\times 2$ matrices, $\frac{1}{\sqrt{|M|}}$
and $\frac{1}{\sqrt{|N|}}$ are to guarantee the consistency, and
$\beta$ and $\gamma$ are constants that can be arbitrary but play
key roles to get a solvable discrete Riccati equation. Equations
\eqref{Lax} can be a Lax pair of equation \eqref{deq} and  $r$
serves as a spectral parameter. A conservation law of \eqref{deq} is
defined by
\begin{equation}
\Delta_m F(u)=\Delta_n J(u), \label{CL-def}
\end{equation}
where $\Delta_n=E_n-1$,  $\Delta_m=E_m-1$ and $u$ solves the
equation \eqref{deq}.

From the Lax pair \eqref{Lax} we construct a formal conservation law
in the following way. First we define
\begin{equation}
\theta=\frac{\t{\phi}_2}{\phi_2},~~\eta=\frac{\h{\phi}_2}{\phi_2}.
\label{theta-eta}
\end{equation}
Then noting that
\begin{equation}
\ln\theta=\Delta_n \ln \phi_2,~~\ln\eta=\Delta_m \ln \phi_2,
\end{equation}
we immediately reach to
\begin{equation}
\Delta_m\ln\theta=\Delta_n \ln\eta, \label{CLs-formal}
\end{equation}
which is a formal conservation law of the lattice equation related
to the Lax pair \eqref{Lax}.

We find that for lattice equations H1, H2, H3, A1, Q1, Q2, and Q3 in
ABS list, $\theta$ satisfies a discrete Riccati equation of the
 following form,\footnote{
In fact, supposing that
$M=\Bigl(\begin{array}{cc}A&B\\C&D\end{array}\Bigr)$ in
\eqref{Lax-a}, one can always get
\begin{equation}
\sqrt{|\t M|}\, \t {\t \phi}_2=\beta \t C \Bigl(\frac{A}{C}-\frac{\t
D}{\t C}\Bigr)\t \phi_2-\beta^2\sqrt{|M|}\frac{\t C}{C}\,\phi_2,
\label{Ri-phi2}
\end{equation}
which, devided by $\phi_2$, yields a discrete Riccati equation of
$\theta$. }
\begin{equation}\label{Ri-ger}
\t\mu \,\t{\theta}\theta=(u-\t{\t{u}})\theta-\varepsilon^{2}\mu,
\end{equation}
where $\mu$ is a function of $u, \t u$ related to considered
equations, $\varepsilon$ is a constant related to $p,r$. It is not
difficult to verify that
\begin{prop}\label{P:1}
The discrete Riccati equation \eqref{Ri-ger} is solved by
\begin{subequations}\label{theta}
\begin{equation}
\theta=\varepsilon^2 \rho
\Bigl(1+\sum_{j=1}^{\infty}\theta_j\varepsilon^{2j}\Bigr),
\end{equation}
with
\begin{align}
& \rho=\frac{\mu}{u-\t{\t{u}}},\\
& \theta_{j+1}=\frac{\t\mu \t
\rho}{u-\t{\t{u}}}\sum_{i=0}^{j}\t{\theta}_i\theta_{j-i},~~
j=0,1,2,\cdots,~~(\theta_0=1).
\end{align}
\end{subequations}
\end{prop}

This gives an explicit form of $\theta$, but it is not enough to get
infinitely many conservation laws from \eqref{CLs-formal}. We still
need an explicit  $\eta$. However, we can not insert $\eta$ into a
Riccati equation similar to \eqref{Ri-ger} with same $\varepsilon$
because $\varepsilon$ is independent of $q$. Fortunately, from the
Lax pair \eqref{Lax} we may find the following relation
\begin{equation}
\eta=\omega(\sigma \theta+1), \label{eta-theat}
\end{equation}
where both $\omega$ and $\sigma$ are functions of $(u,\t u, \h
u,p,q)$ and they satisfy
\begin{equation}
\frac{1}{\omega(u,\t u, \h u,p,q)}=-\sigma(u,\h u, \t u,q,p).
\label{sym}
\end{equation}
Thus, substituting the above $\eta$ together with \eqref{theta} into
the formal conservation law \eqref{CLs-formal}
it is possible to get explicit infinitely many conservation laws. To
do that, we make use of the following expansion formula.
\begin{prop}\label{P:2}
The following expansion holds,
\begin{subequations}\label{exp-t}
\begin{equation}
\ln\biggl(1+\sum_{i=1}^{\infty}t_i k^i
\biggr)=\sum_{j=1}^{\infty}h_j(\bft)k^{j}, \label{ht}
\end{equation}
where
\begin{equation}
h_j(\bft)=\sum_{||\balpha||=j}(-1)^{|\balpha|-1}(|\balpha|-1)!\frac{\bft^{\balpha}}{\balpha
!}, \label{htj}
\end{equation}
and
\begin{align}
& \mathbf{t}=(t_1,t_2,\cdots),~~\balpha=(\alpha_1,\alpha_2,\cdots),~~\alpha_i\in\{0,1,2,\cdots\}, \\
& \bft^{\balpha}=\prod_{i=1}^{\infty}t_i^{\alpha_i},~~
 {\balpha}!=\prod_{i=1}^{\infty}(\alpha_i !),~~
|\balpha|=\sum_{i=1}^{\infty}\alpha_i,~~
||\balpha||=\sum^{\infty}_{i=1} i\alpha_i.
\end{align}
\end{subequations}
The first few of $\{h_j(\bft)\}$ are
\begin{subequations}\label{ht1-4}
\begin{align}
& h_1(\bft)=t_1,\\
& h_2(\bft)=-\frac{1}{2}t_1^2+t_2, \\
& h_3(\bft)=\frac{1}{3}t_1^3-t_1t_2+t_3,\\
& h_4(\bft)=-\frac{1}{4}t_1^4+t_1^2t_2-t_1t_3-\frac{1}{2}t_2^2+t_4.
\end{align}
\end{subequations}
We note that $\{h_j(\bft)\}$ also satisfy
\begin{equation}
\partial_{t_s}h_{i+s}(\bft)=\partial_{t_j}h_{i+j}(\bft), ~~~\mathrm{for}~i,j,s\in \mathbb{Z}^+,
\end{equation}
and $\{h_j(\bft)\}$ are different from the Schur function (See
\cite{Sato-ele}).
\end{prop}

\begin{proof}
Let us first prove the following expansion:
\begin{subequations}
\begin{align}
& \biggl(\sum_{i=1}^{\infty}y_i\biggr)^s=s! \,
\sum_{|\balpha|=s}\frac{\mathbf{y}^{\balpha}}{\alpha !},~~~ (s\in
\mathbb{Z}^+),\label{exp-y}
\\
&
\mathbf{y}=(y_1,y_2,\cdots),~~\bfy^{\balpha}=\prod_{i=1}^{\infty}y_i^{\alpha_i}.
\end{align}
\end{subequations}
Obviously, \eqref{exp-y} is valid for $s=1$. Taking derivative for
\eqref{exp-y} w.r.t. $y_k$ and supposing \eqref{exp-y} is right for
$s-1$, one finds
\begin{equation*}
\frac{\partial}{\partial y_k}l.h.s.~ \mathrm{of}~
\eqref{exp-y}=s\biggl(\sum^{\infty}_{j=1}y_j\biggr)^{s-1}
=s\cdot(s-1)! \sum_{|\balpha|=s-1}\frac{\mathbf{y}^{\balpha}}{\balpha!}
=s! \sum_{|\balpha|=s-1}\frac{\mathbf{y}^{\balpha}}{\balpha!},
\end{equation*}
and meanwhile
\begin{equation*}
\frac{\partial}{\partial y_k}r.h.s.~\mathrm{of}~\eqref{exp-y}
=s!\sum_{|\balpha|=s}\frac{y_k^{\alpha_k-1}}{(\alpha_k-1)!}\prod_{j\neq
k}\frac{y_j^{\alpha_j}}{\alpha_j!}
=s!\sum_{|\balpha|=s-1}\frac{\mathbf{y}^{\balpha}}{\balpha!}.
\end{equation*}
That means \eqref{exp-y} is valid for any $s\in \mathbb{Z}^+$. Then,
noting that
\begin{equation*}
\ln\biggl(1+\sum_{i=1}^{\infty}t_i k^i
\biggr)=\sum_{s=1}^{\infty}(-1)^{s-1}\frac{1}{s}\biggl(\sum_{i=1}^{\infty}t_i
k^i \biggr)^s,
\end{equation*}
using \eqref{exp-y} with $y_i=t_i k^i$, and rearranging the
expansion in terms of $k$, we get \eqref{ht}.
\end{proof}

With this Proposition, from the formal conservation law
\eqref{CLs-formal} we have
\begin{prop}\label{P:3}
When $\theta$ and $\eta$ are defined by \eqref{theta} and
\eqref{eta-theat}, respectively, the formal conservation law
\eqref{CLs-formal} yields infinitely many conservation laws,
\begin{subequations}\label{CLs-inf}
\begin{align}
& \Delta_m \ln \rho=\Delta_n \ln \omega,\\
& \Delta_m\, h_s(\boldsymbol\theta)=\Delta_n\,
h_s(\rho\sigma\underline{\boldsymbol\theta}),~~ (s=1,2,3,\cdots),
\end{align}
where $h_s(\bft)$ is defined in \eqref{htj},
\begin{equation}
\boldsymbol\theta=(\mathbf{\theta}_1,\theta_2,\cdots),~~~
\underline{\boldsymbol\theta}=(1,\theta_1,\theta_2,\cdots),
\end{equation}
and $\rho\sigma\underline{\boldsymbol\theta}=(\sigma\rho,\sigma\rho\theta_1,\sigma\rho\theta_2,\cdots)$.
\end{subequations}

\end{prop}

\section{Conservation laws of ABS lattice equations}

\subsection{Conservation laws for H1, H2, H3, Q1, Q2, Q3 and A1 equation}

Let us first, taking H1 equation as an example, give a detailed
procedure of deriving infinitely many conservation laws.

The Lax pair of H1 equation reads
\begin{subequations}\label{Lax-H1}
\begin{align}\label{phi-t-H1}
& \t{\phi}= \frac{\beta}{\sqrt{r-p}}\left(\begin{array}{cc}
      u     & -u\t{u}+p-r\\
      1     & -\t{u}
      \end{array}
\right)\phi,\\
\label{phi-h-H1} & \h{\phi}=\frac{\gamma}{\sqrt{r-q}}\left(
\begin{array}{cc}
u & -u\h{u}+q-r\\
1 & -\h{u}
\end{array} \right)\phi.
\end{align}
\end{subequations}
Taking $\beta=\sqrt{r-p}=\varepsilon$, from \eqref{phi-t-H1} we can
find
\begin{equation}
\t{\t{\phi}}_2=(u-\t{\t{u}}) \t{\phi}_{2}-\varepsilon^{2}\phi_{2},
\label{tt-H1}
\end{equation}
which leads to
\begin{equation}
\t{\theta}\theta=(u-\t{\t{u}})\theta-\varepsilon^{2}, \label{Ri-H1}
\end{equation}
where $\theta$ is defined in \eqref{theta-eta}, i.e.,
$\theta=\t{\phi}_{2}/\phi_{2}$. This is the discrete Riccati
equation \eqref{Ri-ger} with $\mu=1$ and is solved by \eqref{theta}.
Further, taking $\gamma=\sqrt{r-q}$, from the Lax pair
\eqref{Lax-H1} we have
\begin{equation}
\theta=\frac{\t{\phi}_2}{\phi_2}=\frac{\phi_1}{\phi_2}-\t{u},~~
\eta=\frac{\h{\phi}_2}{\phi_2}=\frac{\phi_1}{\phi_2}-\h{u},
\end{equation}
which by eliminating $\phi_1/\phi_2$ yields the relation
\begin{equation}
\eta=\theta+\t{u}-\h{u}.
\end{equation}
This is \eqref{eta-theat} with $\omega=\t{u}-\h{u}$ and $\sigma=1/(\t{u}-\h{u})$. Thus, for H1 equation, based on Proposition
\ref{P:3} we can write out infinitely many conservation laws
\eqref{CLs-inf} with $\mu,~ \omega,~ \sigma$ obtained above.
We note that these conservation laws are as same as those derived via Gardner method\cite{rs-2009}

For the lattice equations H2, H3, A1, Q1, Q2, and Q3, starting from
their Lax pairs, we can also derive infinitely many conservation
laws through a similar procedure. Let us skip details and list main
results of these equations together with H1 equation.

\begin{prop}\label{P:4}
For the lattice equations H1, H2, H3, A1, Q1, Q2, and Q3 in ABS
list, starting from their Lax pairs, one can construct a formal
conservation law \eqref{CLs-formal} with $\theta$ and $\eta$ defined
in \eqref{theta-eta}, where $\theta$ satisfies the discrete Riccati
equation \eqref{Ri-ger} solved by \eqref{theta} and $\eta$ is
expressed through \eqref{eta-theat}. By means of the polynomials
$\{h_j(\bft)\}$ defined in \eqref{htj}, one can explicitly express
the infinitely many conservation laws as \eqref{CLs-inf}. In the
following for H1, H2, H3, A1, Q1, Q2, and Q3 equation, we list out
their Lax pairs, parametrisation of $\beta,~\gamma$, and auxiliary
functions $\mu,~ \omega$ and $\sigma$. For H1 equation, its Lax pair
reads \eqref{Lax-H1} with
\begin{equation}
\beta=\sqrt{r-p}=\varepsilon,~~\gamma=\sqrt{r-q},~~\mu=1,~~
\omega=\t{u}-\h{u},~~ \sigma=\frac{1}{\t{u}-\h{u}}.
\end{equation}
For H2 equation, its Lax pair reads
\begin{subequations}
\begin{align}
\label{phi-t-H2} & \t{\phi}=\frac{\beta}{\sqrt{2(r-p)(p+u+\t{u})}}
\left(
\begin{array}{cc}
u+p-r & -u\t{u}+(p-r)(u+\t{u})+p^2-r^2\\
1 & -\t{u}-p+r
\end{array}
\right)\phi,\\
\label{phi-h-H2} &
\h{\phi}=\frac{\gamma}{\sqrt{2(r-q)(q+u+\h{u})}}\left(
\begin{array}{cc}
u+q-r & -u\h{u}+(q-r)(u+\h{u})+q^2-r^2\\
1 & -\h{u}-q+r
\end{array}
\right)\phi,
\end{align}
\end{subequations}
and
\begin{subequations}
\begin{align}
& \beta=\sqrt{2(r-p)}=\varepsilon,~~\gamma=\sqrt{2(r-q)},\\
& \mu=\sqrt{p+u+\t{u}},~~
\omega=\frac{p-q+\t{u}-\h{u}}{\sqrt{q+u+\h{u}}},~~
\sigma=\frac{\sqrt{p+u+\t{u}}}{p-q+\t{u}-\h{u}}.
\end{align}
\end{subequations}
For H3 equation, its Lax pair reads
\begin{subequations}
\begin{align}
\label{phi-t-H3} &
\t{\phi}=\frac{\beta}{\sqrt{(p^2-r^2)(u\t{u}+p\delta)}} \left(
\begin{array}{cc}
ru & -pu\t{u}-\delta(p^2-r^2)\\
p & -r\t{u}
\end{array}
\right)\phi,\\
\label{phi-h-H3} &
\h{\phi}=\frac{\gamma}{\sqrt{(q^2-r^2)(u\h{u}+q\delta)}}\left(
\begin{array}{cc}
ru & -qu\h{u}-\delta(q^2-r^2)\\
q & -r\h{u}
\end{array}
\right)\phi,
\end{align}
\end{subequations}
and
\begin{subequations}
\begin{align}
& \beta=\frac{\sqrt{p^2-r^2}}{r}=\varepsilon,~~\gamma=\frac{\sqrt{q^2-r^2}}{r},\\
& \mu=\sqrt{u\t{u}+p\delta},~~
\omega=\frac{q\t{u}-p\h{u}}{p\sqrt{u\h{u}+q\delta}},~~
\sigma=\frac{q\sqrt{u\t{u}+p\delta}}{q\t{u}-p\h{u}}.
\end{align}
\end{subequations}
For Q1 equation, its Lax pair reads
\begin{subequations}
\begin{align}
\label{phi-t-Q1} &
\t{\phi}=\frac{\beta}{\sqrt{r(r-p)[(u-\t{u})^2-\delta^2p^2]}}
\left(\begin{array}{cc} ru+(p-r)\t{u} & -pu\t{u}-\delta^2pr(p-r)\\
p & (r-p)u-r\t{u}
\end{array}
\right)\phi,\\
\label{phi-h-Q1} &
\h{\phi}=\frac{\gamma}{\sqrt{r(r-q)[(u-\h{u})^2-\delta^2q^2]}}\left(
\begin{array}{cc}
ru+(q-r)\h{u} & -qu\h{u}-\delta^2qr(q-r)\\
q & (r-q)u-r\h{u}
\end{array}
\right)\phi,
\end{align}
\end{subequations}
and
\begin{subequations}
\begin{align}
& \beta=\frac{\sqrt{r(r-p)}}{r}=\varepsilon,~~\gamma=\frac{\sqrt{r(r-q)}}{r},\\
& \mu=\sqrt{(u-\t{u})^2-\delta^2p^2},~~
\omega=\frac{q(\t{u}-u)-p(\h{u}-u)}{p\sqrt{(u-\h{u})^2-\delta^2q^2}},
~~
\sigma=\frac{q\sqrt{(u-\t{u})^2-\delta^2p^2}}{q(\t{u}-u)-p(\h{u}-u)}.
\end{align}
\end{subequations}
For Q2 equation, its Lax pair reads
\begin{subequations}
\begin{align}
\label{phi-t-Q2} & \t{\phi}=\frac{\beta}{\sqrt{A}}
\left(\begin{array}{cc} ru+(p-r)\t{u}-pr(p-r) & -pu\t{u}-pr(p-r)(u+\t{u}-p^2+pr-r^2)\\
p & (r-p)u-r\t{u}+pr(p-r)
\end{array}
\right)\phi,\\
\label{phi-h-Q2} & \h{\phi}=\frac{\gamma}{\sqrt{B}}\left(
\begin{array}{cc} ru+(q-r)\h{u}-qr(q-r) & -qu\h{u}-qr(q-r)(u+\h{u}-q^2+qr-r^2)\\
q & (r-q)u-r\h{u}+qr(q-r)
\end{array}
\right)\phi,
\end{align}
with \begin{align}
A=r(r-p)[(p^2-u)^2+\t{u}(\t{u}-2u-2p^2)],~~
 B=r(r-q)[(q^2-u)^2+\h{u}(\h{u}-2u-2q^2)],
 \end{align}
\end{subequations}
and we take
\begin{subequations}
\begin{align}
&
\beta=\frac{\sqrt{r(r-p)}}{r}=\varepsilon,~~\gamma=\frac{\sqrt{r(r-q)}}{r},~~
\mu=\sqrt{(p^2-u)^2+\t{u}(\t{u}-2u-2p^2)},\\
&
\omega=\frac{(p-q)(u-pq)+q\t{u}-p\h{u}}{p\sqrt{(q^2-u)^2+\h{u}(\h{u}-2u-2q^2)}},
~~
\sigma=\frac{q\sqrt{(p^2-u)^2+\t{u}(\t{u}-2u-2p^2)}}{(p-q)(u-pq)+q\t{u}-p\h{u}}.
\end{align}
\end{subequations}
For Q3 equation, its Lax pair reads
\begin{subequations}
\begin{align}
\label{phi-t-Q3} & \t{\phi}=\frac{\beta}{\sqrt{A}}
\left(\begin{array}{cc} p(r^2-1)u-(r^2-p^2)\t{u} &~~ -r(p^2-1)u\t{u}+\delta^2(p^2-r^2)(p^2-1)(r^2-1)/4pr\\
r(p^2-1) & (r^2-p^2)u-p(r^2-1)\t{u}
\end{array}
\right)\phi,\\
\label{phi-h-Q3} & \h{\phi}=\frac{\gamma}{\sqrt{B}}
\left(\begin{array}{cc} q(r^2-1)u-(r^2-q^2)\h{u} &~~ -r(q^2-1)u\h{u}+\delta^2(q^2-r^2)(q^2-1)(r^2-1)/4qr\\
r(q^2-1) & (r^2-q^2)u-q(r^2-1)\h{u}
\end{array}
\right)\phi,
\end{align}
with
\begin{align}
& A=(r^2-1)(r^2-p^2) [(u-p \t u)(p u-\t u)+\delta^2(1-p^2)^2/(4p) ],\\
& B=(r^2-1)(r^2-q^2) [(u-q \h u)(q u-\h
u)+\delta^2(1-q^2)^2/(4q) ],
\end{align}
\end{subequations}
and we take
\begin{subequations}
\begin{align}
& \beta=\frac{\sqrt{(r^2-1)(r^2-p^2)}}{r^2-1}=\varepsilon,~~\gamma=\frac{\sqrt{(r^2-1)(r^2-q^2)}}{r^2-1},\\
& \mu=\sqrt{(u-p \t u)(p u-\t u)+\delta^2(1-p^2)^2/(4p)},\\
& \omega=\frac{p(q^2-1)\t{u}+(p^2-q^2)u-q(p^2-1)\h{u}}{(p^2-1)\sqrt{(u-q \h u)(q u-\h u)+\delta^2(1-q^2)^2/(4q)}},\\
& \sigma=\frac{(q^2-1)\sqrt{(u-p \t u)(p u-\t
u)+\delta^2(1-p^2)^2/(4p)}}{p(q^2-1)\t{u}+(p^2-q^2)u-q(p^2-1)\h{u}}.
\end{align}
\end{subequations}
For A1 equation, its Lax pair reads
\begin{subequations}
\begin{align}
\label{phi-t-A1} &
\t{\phi}=\frac{\beta}{\sqrt{r(p-r)\bigl[(u+\t{u})^2-\delta^2p^2\bigr]}}
\left(\begin{array}{cc} ru+(r-p)\t{u} & -pu\t{u}+\delta^2pr(p-r)\\
p & (p-r)u-r\t{u}
\end{array}
\right)\phi,\\
\label{phi-h-A1} &
\h{\phi}=\frac{\gamma}{\sqrt{r(q-r)\bigl[(u+\h{u})^2-\delta^2q^2\bigr]}}\left(
\begin{array}{cc}
ru+(r-q)\h{u} & -qu\h{u}+\delta^2qr(q-r)\\
q & (q-r)u-r\h{u}
\end{array}
\right)\phi,
\end{align}
\end{subequations}
and
\begin{subequations}
\label{A1-para}
\begin{align}
& \beta=\frac{\sqrt{r(p-r)}}{r}=\varepsilon,~~\gamma=\frac{\sqrt{r(q-r)}}{r},\\
& \mu=\sqrt{(u+\t{u})^2-\delta^2p^2},~~
 \omega=\frac{q(u+\t{u})-p(u+\h{u})}{p\sqrt{(u+\h{u})^2-\delta^2q^2}},
~~
\sigma=\frac{q\sqrt{(u+\t{u})^2-\delta^2p^2}}{q(u+\t{u})-p(u+\h{u})}.
\end{align}
\end{subequations}
\end{prop}

There is a transformation\cite{ABS-2003}
\begin{equation}
u=(-1)^{n+m}v \label{trans-Q1-A1}
\end{equation}
connecting Q1 equation and A1 equation
\begin{equation}
p(v+\h{v})(\t{v}+\th{v})-q(v+\t{v})(\h{v}+\th{v})-\delta^2pq(p-q)=0.
\end{equation}
In fact, one can substitute \eqref{trans-Q1-A1} into the infinitely
many conservation laws of Q1 equation to get those of A1 equation.
The obtained conservation laws are as same as those derived through
\eqref{A1-para} (with $v$ in place of $u$).

\subsection{Conservation laws for A2 equation}
\subsubsection{Transformation}

The transformation\cite{ABS-2003}
\begin{equation}
u=v^{(-1)^{n+m}} \label{trans-Q3-A2}
\end{equation}
connects Q3$|_{\delta=0}$ equation
\begin{equation}
(q^2-p^2)(u\th{u}+\t{u}\h{u})+q(p^2-1)(u\t{u}+\h{u}\th{u})-p(q^2-1)(u\h{u}+\t{u}\th{u})=0
\label{Q30}
\end{equation}
and A2 equation
\begin{equation}
(q^2-p^2)(v\t{v}\h{v}\th{v}+1)+q(p^2-1)(v\h{v}+\t{v}\th{v})-p(q^2-1)(v\t{v}+\h{v}\th{v})=0.
\label{A2-v}
\end{equation}
Noting that the conservation law \eqref{CL-def} of equation
\eqref{deq} is a relation that holds for all of $u$ satisfying
\eqref{deq}, for A2 equation \eqref{A2-v} what we need  is to list
out conservation laws of Q3$|_{\delta=0}$ equation \eqref{Q30} and
then replace  $u$ by $v^{(-1)^{n+m}}$.

\begin{prop}\label{P:5}
The infinitely many conservation laws of Q3$|_{\delta=0}$ equation
\eqref{Q30} is given by \eqref{CLs-inf} with $\theta$, $\eta$ and
$\{h_j(\bft)\}$ given in \eqref{theta}, \eqref{eta-theat} and
\eqref{htj}, respectively, and
\begin{subequations}
\begin{align}
& \mu=\sqrt{(u-p \t u)(p u-\t u)},\\
&
\omega=\frac{p(q^2-1)\t{u}+(p^2-q^2)u-q(p^2-1)\h{u}}{(p^2-1)\sqrt{(u-q\h
u)(qu-\h u)}}, ~~ \sigma=\frac{(q^2-1)\sqrt{(u-p \t u)(p u-\t
u)}}{p(q^2-1)\t{u}+(p^2-q^2)u-q(p^2-1)\h{u}}.
\end{align}
\end{subequations}
The infinitely many conservation laws of A2 equation \eqref{A2-v} can be
given through the infinitely many conservation laws of Q3$|_{\delta=0}$
equation \eqref{Q30} by replacing   $u$ by $v^{(-1)^{n+m}}$.
\end{prop}

\subsubsection{Lax pair approach}\label{sec}
Conservation laws of A2 equation can also be derived directly from
its Lax pair which reads
\begin{subequations}\label{Lax-A2}
\begin{align}
\label{phi-t-A2} & \t{\phi}=\frac{\beta}{\sqrt{A}}
\left(\begin{array}{cc}  -r(p^2-1)u & p(r^2-1)u\t{u}-(r^2-p^2) \\
(r^2-p^2)u\t{u}-p(r^2-1) & r(p^2-1)\t{u}
\end{array}
\right)\phi,\\
\label{phi-h-A2} & \h{\phi}=\frac{\gamma}{\sqrt{B}}
\left(\begin{array}{cc}  -r(q^2-1)u & q(r^2-1)u\h{u}-(r^2-q^2) \\
(r^2-q^2)u\h{u}-q(r^2-1) & r(q^2-1)\h{u}
\end{array}
\right)\phi,
\end{align}
\end{subequations}
with
\[ A=(r^2-1)(r^2-p^2)(p-u\t u)(p u\t u-1),~~
 B=(r^2-1)(r^2-q^2)(q-u\h{u})(q u\h{u}-1).\]
Since, in this case, the equation \eqref{Ri-phi2} becomes
\begin{equation}
\sqrt{\t A}\, \t{\t \phi}_2=\beta \frac{ r p (p^2-1)(r^2-1)(u-\t {\t
u})}{p(1-r^2)+(r^2-p^2)u\t u} \t \phi_2 -\beta^2\sqrt{A}\,
\frac{p(1-r^2)+(r^2-p^2)\t u\t {\t u}}{p(1-r^2)+(r^2-p^2)u\t u}
\phi_2,
\end{equation}
which is difficult to get a solvable Riccati equation for
$\theta=\t{\phi}_2/\phi_2$, we turn to another formulae set
\begin{subequations}
\label{theta-eta-A2}
\begin{align}
\label{theta-A2} &
\theta=\frac{\t{\phi}_2}{\phi_2}=\frac{1}{\mu}\bigl[\bigl[(r^2-p^2)u\t{u}-p(r^2-1)\bigr]\zeta+r(p^2-1)\t{u}\bigr],
\\
\label{eta-A2} &
\eta=\frac{\h{\phi}_2}{\phi_2}=\frac{1}{\nu}\bigl[\bigl[(r^2-q^2)u\h{u}-q(r^2-1)\bigr]\zeta
+r(q^2-1)\h{u}\bigr].
\end{align}
By $\theta$ and $\eta$ the formal conservation law is written as
\begin{equation*}
\Delta_m \ln \theta=\Delta_n \ln \eta.
\label{cls-fromal-A2}
\end{equation*}
\eqref{theta-A2} and \eqref{eta-A2} are derived from the Lax pair \eqref{Lax-A2}, where
\begin{equation}
\zeta=\frac{\phi_1}{\phi_2},
\end{equation}
and
\begin{align}\label{mu-nu-A2}
  \mu=\sqrt{(p-u\t u)(p u\t u-1)},~~
\nu=\sqrt{(q-u\h{u})(q u\h{u}-1)},
\end{align}
and we have taken
\begin{align*}
\beta=\sqrt{(r^2-1)(r^2-p^2)},~~ \gamma =\sqrt{(r^2-1)(r^2-q^2)}.
\end{align*}
\end{subequations}
$\zeta$ is determined by the following equation
\begin{equation}\label{bt-A2-1}
\t{\zeta}=\frac{-r(p^2-1)u\, \zeta+p(r^2-1)u\t{u}-(r^2-p^2)}
{\bigl[(r^2-p^2)u\t{u}-p(r^2-1)\bigr]\zeta+r(p^2-1)\t{u}},
\end{equation}
which is derived from \eqref{phi-t-A2}.\\

 To solve \eqref{bt-A2-1},
we take (cf.\cite{Gardner-1})
\begin{equation}
\zeta=\t u+\xi, ~~ \varepsilon= r-p,
\label{zeta-A2}
\end{equation}
and we reach to
\begin{align}
(a_0+a_1\varepsilon+a_2\varepsilon^2)\xi\t \xi+
(b_1\varepsilon+b_2\varepsilon^2)\t \xi+
(c_0+c_1\varepsilon+c_2\varepsilon^2)\xi +
(d_1\varepsilon+d_2\varepsilon^2)=0, \label{Ri-xi-A2}
\end{align}
\begin{subequations}
where
\begin{align*}
& a_0=-p(p^2-1),~~ a_1=2p(u\t{u}-p),~~ a_2=u\t{u}-p,\\
& b_1=(2pu\t{u}-p^2-1)\t{u},~~ b_2=(u\t{u}-p)\t{u},\\
& c_0=p(p^2-1)(u-\t{\t{u}}),~~c_1=2p\t{\t{u}}(u\t{u}-p)+(p^2-1)u,~~c_2=(u\t{u}-p)\t{\t{u}},\\
&
d_1=2p(u\t{u}^2\t{\t{u}}+1)-(p^2+1)(u\t{u}+\t{u}\t{\t{u}}),~~d_2=u\t{u}^2\t{\t{u}}+1-p(u\t{u}+\t{u}\t{\t{u}}).
\end{align*}
\end{subequations}
Equation \eqref{Ri-xi-A2} is then solved by
\begin{subequations}
\[
\xi=\sum^{\infty}_{j=1}\xi_j \varepsilon^j,
\]
with
\begin{align*}
& \xi_1=-\frac{d_1}{c_0},\\
& \xi_2=-\frac{1}{c_0}(a_0\xi_1\t \xi_1+b_1\t \xi_1+c_1 \xi_1+d_2),\\
&
\xi_3=-\frac{1}{c_0}\bigl[a_0(\xi_1\t{\xi}_2+\xi_2\t{\xi}_1)+a_1\xi_1\t{\xi}_1+b_1\t{\xi}_2+b_2\t{\xi}_1+c_1\xi_2+c_2\xi_1\bigr],\\
& \xi_s=-\frac{1}{c_0}\Bigl( \sum^{2}_{k=0}a_k\sum^{s-k-1}_{i=1}
\xi_i \t \xi_{s-k-i}+ \sum^{2}_{k=1}b_k\t \xi_{s-k}+
\sum^{2}_{k=1}c_k \xi_{s-k} \Bigr),~~~ (s=4,5,\cdots).
\end{align*}
\end{subequations}\\
Next, we express \eqref{theta-eta-A2} in terms of $\xi$ and
$\varepsilon$ as
\begin{subequations}
\label{theta-eta-A2-xi}
\begin{align}
\label{theta-A2-xi} &
\theta=\frac{1}{\mu}\bigl[(f_0+f_1\varepsilon+f_2\varepsilon^2)\xi+g_1\varepsilon+g_2\varepsilon^2\bigr],
\\
\label{eta-A2-xi} &
\eta=\frac{1}{\nu}\bigl[(w_0+w_1\varepsilon+w_2\varepsilon^2)\xi +
z_0+z_1\varepsilon+z_2\varepsilon^2\bigr],
\end{align}
with
\begin{align*}
&f_0=p(1-p^2),~~f_1=2p(u\t{u}-p),~~f_2=u\t{u}-p,\\
&g_1=(2pu\t{u}-p^2-1)\t{u},~~g_2=(u\t{u}-p)\t{u},\\
&w_0=(p^2-q^2)u\h u -q(p^2-1),~~w_1=2p(u\h u-q),~~w_2=u\h u-q,\\
&z_0=(p^2-q^2)u\h u\t u-q(p^2-1)\t u + p(q^2-1)\h u,~~z_1=2p\t u(u\h
u-q)+(q^2-1)\h u,~~z_2=(u\h u-q)\t u,
\end{align*}
\end{subequations}
and we can get solutions
\begin{equation}
\theta=\frac{(f_0\xi_1+g_1)\varepsilon}{\mu}
\biggl(1+\sum^{\infty}_{j=1}\theta_j\varepsilon^j\biggr),
\end{equation}
with
\begin{subequations}
\label{theta-j-A2}
\begin{align}
&\theta_1=\frac{f_0\xi_2+f_1\xi_1+g_2}{f_0\xi_1+g_1},\\
&\theta_s=\frac{f_0\xi_{s+1}+f_1\xi_s+f_2\xi_{s-1}}{f_0\xi_1+g_1},~~(s=2,3,\cdots),
\end{align}
\end{subequations}
and
\begin{equation}\eta=\frac{z_0}{\nu}\biggl(1+\sum^{\infty}_{j=1}\eta_j\varepsilon^j\biggr),
\end{equation}
with
\begin{subequations}
\label{eta-j-A2}
\begin{align}
&\eta_1=\frac{1}{z_0}(w_0\xi_1+z_1),\\
&\eta_2=\frac{1}{z_0}(w_0\xi_2+w_1\xi_1+z_2),\\
&\eta_s=\frac{1}{z_0}\sum\limits_{i=0}^{2}w_i\xi_{s-i},~~(s=3,4,\cdots).
\end{align}
\end{subequations}

Finally, by means of the polynomials $\{h_j(\bft)\}$ defined in \eqref{htj}, the infinitely
many conservation laws of A2 equation are given by
\begin{subequations}
\begin{align}
&\Delta_m \ln \frac{f_0\xi_1+g_1}{\mu}=\Delta_n \ln
\frac{z_0}{\nu},\\
&\Delta_m h_s(\boldsymbol\theta)=\Delta_n
h_s(\boldsymbol{\eta}),~~s=1,2,\cdots,
\end{align}
where
\[\boldsymbol\theta=(\theta_1,\theta_2,\cdots),~~~\boldsymbol\eta=(\eta_1,\eta_2,\cdots),\]
with $\mu,$ $\nu,$ $\{\theta_j\}$ and $\{\eta_j\}$ given in
\eqref{mu-nu-A2}, \eqref{theta-j-A2} and \eqref{eta-j-A2}
respectively.
\end{subequations}

\subsection{Conservation laws for Q4 equation}

For Q4 equation, one can use the same method as in Sec.\ref{sec}. Lax
pair of Q4 equation is
\begin{subequations}\label{Lax-Q4}
\begin{align}\label{phi-t-Q4}
 &\t{\phi}=\frac{\beta}{\sqrt{A}}\left(
\begin{array}{cc}
r(1-p^2r^2)u+(pR-rP)\t u & -p(1-p^2r^2)u\t u-pr(pR-rP)\\
p(1-p^2r^2)+pr(pR-rP)u\t u &-r(1-p^2r^2)\t u-(pR-rP)u
\end{array}
\right)\phi,\\
&\label{phi-h-Q4} \h{\phi}=\frac{\gamma}{\sqrt{B}}\left(
\begin{array}{cc}
r(1-q^2r^2)u+(qR-rQ)\h u & -q(1-q^2r^2)u\h u-qr(qR-rQ)\\
q(1-q^2r^2)+qr(qR-rQ)u\h u &-r(1-q^2r^2)\h u-(qR-rQ)u
\end{array}
\right)\phi,
\end{align}
with
\begin{align*}
&A=r(1-p^2r^2)(pR-rP)\bigl[2Pu\t{u}+p^2(u^2\t{u}^2+1)-\t{u}^2-u^2\bigr],\\
&B=r(1-q^2r^2)(qR-rQ)\bigl[2Qu\h{u}+q^2(u^2\h{u}^2+1)-\h{u}^2-u^2\bigr],
\end{align*}
and $(r,R)$ are formulated by the elliptic curve
\begin{equation}
R^2=r^4- k r^2+1.
\label{R-r-A2}
\end{equation}
\end{subequations}
Taking $\beta=\sqrt{r(1-p^2r^2)(pR-rP)}$ and
$\gamma=\sqrt{r(1-q^2r^2)(qR-rQ)}$ in \eqref{Lax-Q4},  we have
\begin{subequations}\label{theta-eta-Q4}
\begin{align}
&\theta=\frac{1}{\mu}\bigl[\bigl[p(1-p^2r^2)+pr(pR-rP)u\t
u\bigr]\zeta
-r(1-p^2r^2)\t u-(pR-rP)u\bigr],\\
&\eta=\frac{1}{\nu}\bigl[\bigl[q(1-q^2r^2)+qr(qR-rQ)u\h u\bigr]\zeta
-r(1-q^2r^2)\h u-(qR-rQ)u\bigr],
\end{align}
where
\begin{equation}
\zeta=\frac{\phi_1}{\phi_2}, \label{zeta-Q4}
\end{equation}
 and
\begin{equation}\label{mu-nu-Q4}
\mu=\sqrt{2Pu\t{u}+p^2(u^2\t{u}^2+1)-\t{u}^2-u^2},~~\nu=\sqrt{2Qu\h{u}+q^2(u^2\h{u}^2+1)-\h{u}^2-u^2}.
\end{equation}
\end{subequations}
The formal conservation law is given by
\begin{equation}
\Delta_m \ln \theta=\Delta_n \ln \eta.
\label{cls-fromal-Q4}
\end{equation}
$\zeta$ is determined by the following Riccati equation
\begin{equation}\label{bt-Q4-1}
\t \zeta(c\zeta +d)=a\zeta +b,
\end{equation}
which is derived from \eqref{phi-t-Q4}, and
here
\begin{align*}
&a=r(1-p^2r^2)u+(pR-rP)\t u,~~b=-p(1-p^2r^2)u\t u-pr(pR-rP),\\
&c=p(1-p^2r^2)+pr(pR-rP)u\t u,~~d=-r(1-p^2r^2)\t u-(pR-rP)u.
\end{align*}
To solve it we take (cf.\cite{Gardner-1})
\begin{equation}
\zeta=\t u+\xi,~~\varepsilon=r-p,
\label{zeta-xi-Q4}
\end{equation}
and expand $a, b, c, d$ as
\begin{equation}
a=\sum^{\infty}_{i=0}a_i\varepsilon^i,~b=\sum^{\infty}_{i=0}b_i\varepsilon^i,
~c=\sum^{\infty}_{i=0}c_i\varepsilon^i,~d=\sum^{\infty}_{i=0}d_i\varepsilon^i.
\label{abcd-Q4}
\end{equation}
Since $R$ defined in \eqref{R-r-A2} can be expanded as
\begin{equation}
R=\sum^{\infty}_{i=0}r_i\varepsilon^i,
\end{equation}
in which $r_i$ is given by
\begin{align*}
r_i=P\sum\limits_{||\balpha||=i}\frac{\bfg^{\balpha}}{\balpha
!}\prod_{i=0}^{|\balpha|-1}(\frac{1}{2}-i),
\end{align*}
where
\begin{align*}
&\bfg=(g_1,g_2,g_3,g_4),~~\balpha=(\alpha_1,\alpha_2,\alpha_3,\alpha_4),~~\alpha_i\in\{0,1,2,\cdots\}\\
&g_1=\frac{2}{P}(2p^3- k p),~g_2=\frac{1}{P}(6p^2- k),~g_3=\frac{4p}{P},~g_4=\frac{1}{P},\\
&||\balpha||=\sum_{i=1}^{4}j\alpha_j,~~|\balpha|=\sum_{i=1}^{4}\alpha_i,~~\balpha!=\prod_{i=1}^{4}(\alpha_i!),
~~\bfg^{\balpha}=\prod_{i=1}^4g_i^{\alpha_i},
\end{align*}
we have
\begin{align*}
&a_0=p(1-p^4)u,~~a_1=(1-3p^4)u+(pr_1-P)\t u,~~a_2=-3p^3u+pr_2\t u,~~ a_3=-p^2u+pr_3\t u,\\
& a_i=pr_i\t u,~~(i=4,5,\cdots),\\
&b_0=-p(1-p^4)u\t u,~~b_1=2p^4u\t u-p^2(pr_1-P),~~b_2=p^3u\t u-p(pr_1-P)-p^3r_2,\\
& b_i=-p^2(pr_i+r_{i-1}),~~(i=3,4,\cdots),\\
&c_0=p(1-p^4),~~c_1=p^2(pr_1-P)u\t u-2p^4,~~c_2=\bigl[p(pr_1-P)+p^3r_2\bigr]u\t u-p^3,\\
& c_i=p^2(pr_i+r_{i-1})u\t u,~~(i=3,4,\cdots),\\
&d_0=-p(1-p^4)\t u,~~d_1=-(1-3p^4)\t u-(pr_1-P)u,~~d_2=3p^3\t u-pr_2u,~~ d_3=p^2\t u-pr_3u,\\
&d_i=-pr_iu,~~(i=4,5,\cdots).
\end{align*}
Consequently, \eqref{bt-Q4-1} turns out to be
\begin{align}\label{Ri-xi-Q4}
\sum_{i=0}^{\infty}c_i\varepsilon^i\xi\t \xi
+\sum_{i=1}^{\infty}(c_i\t u+d_i)\varepsilon^i\t \xi
+\sum_{i=0}^{\infty}(c_i\t{\t{u}}-a_i)\varepsilon^i\xi
+\sum_{i=1}^{\infty}\bigl[(c_i\t u +d_i)\t{\t{u}}-a_i\t
u-b_i\bigr]\varepsilon^i=0.
\end{align}
This is solved by
\begin{subequations}
\begin{align*}
\xi=\sum^{\infty}_{j=1}\xi_j \varepsilon^j,
\end{align*}
with
\begin{align*}
&\xi_1=\frac{-1}{c_0\t{\t{u}}-a_0}\bigl[(c_1\t u
+d_1)\t{\t{u}}-a_1\t
u-b_1\bigr],\\
&\xi_s=\!\frac{-1}{c_0\t{\t{u}}-a_0}
\Bigl[\sum\limits_{k=0}^{s-2}\sum_{i=1}^{s-k-1}\!\!\!c_k\xi_i\t
\xi_{s-k-i}\!+\!\!\sum\limits_{i=1}^{s-1}\!\big[(c_i\t u+d_i)\t
\xi_{s-i}\!+\!(c_i\t{\t{u}}-a_i)\xi_{s-i}\big]\!\!+\!(c_s\t u
+d_s)\t{\t{u}}\!-a_s\t u\!-b_s\Bigr],
\end{align*}
\end{subequations}
for $s=2,3,\cdots$.

 Next, from \eqref{zeta-xi-Q4} and \eqref{theta-eta-Q4} one has
\begin{subequations}
\begin{align}
&\theta=\frac{1}{\mu}\biggl(\sum_{i=0}^{\infty}c_i\varepsilon^i\xi+\sum_{i=0}^{\infty}f_i\varepsilon^i\biggr),\\
&\eta=\frac{1}{\nu}\biggl(\sum_{i=0}^{\infty}w_i\varepsilon^i\xi+\sum_{i=0}^{\infty}z_i\varepsilon^i\biggr),
\end{align}
where
\begin{align*}
&f_0=0,~~f_i=c_i\t u+d_i,~~(i=1,2,\cdots),\\
&w_0=q(1-q^2p^2)+qp(qP-pQ)u\h u,~~w_1=-2pq^3+q(qP+pqr_1-2pQ)u\h
u,\\
&w_2=-q^3+q(qpr_2+qr_1-Q)u\h u,~~w_i=q^2(pr_i+r_{i-1})u\h
u,~~(i=3,4,\cdots),\\
&z_0=-p(1-q^2p^2)\h u-(qP-pQ)u+w_0\t u,~~z_1=(3p^2q^2-1)\h u-(qr_1-Q)u+w_1\t u,\\
&z_2=3pq^2\h u-qr_2u+w_2\t u,~~z_3=q^2\h u-qr_3u+w_3\t
u,~~z_i=-qr_iu+w_i\t u,~~(i=4,5,\cdots).
\end{align*}
\end{subequations}
Further
\begin{equation}
\theta=\frac{(c_0\xi_1+f_1)\varepsilon}{\mu}\biggl(1+\sum_{j=1}^{\infty}\theta_j\varepsilon^j\biggr),
\end{equation}
where
\begin{equation}
\theta_j=\frac{1}{c_0\xi_1+f_1}\biggl(\sum_{k=0}^jc_k\xi_{j+1-k}+f_{j+1}\biggr),~~(j=1,2,\cdots),
\label{theta-j-Q4}
\end{equation}
and
\begin{equation}
\eta=\frac{z_0}{\nu}\biggl(1+\sum_{j=1}^{\infty}\eta_j\varepsilon^j\biggr),
\end{equation}
where
\begin{equation}
\eta_j=\frac{1}{z_0}\biggl(\sum_{k=0}^{j-1}w_k\xi_{j-k}+z_j\biggr),~~(j=1,2,\cdots).
\label{eta-j-Q4}
\end{equation}
Finally, by means of the polynomials $\{h_j(\bft)\}$ defined in
\eqref{htj}, from the formal conservation law \eqref{cls-fromal-Q4},
the infinitely many conservation laws of Q4 equation are given by
\begin{subequations}
\begin{align}
&\Delta_m \ln \frac{c_0\xi_1+f_1}{\mu}=\Delta_n \ln
\frac{z_0}{\nu},\\
&\Delta_m h_s(\boldsymbol\theta)=\Delta_n
h_s(\boldsymbol\eta),~~~s=1,2,\cdots,
\end{align}
where
\[\boldsymbol\theta=(\theta_1,\theta_2,\cdots),~~~\boldsymbol\eta=(\eta_1,\eta_2,\cdots),\]
with $\mu$, $\nu$, $\theta_j$ and $\eta_j$ are given by
\eqref{mu-nu-Q4}, \eqref{theta-j-Q4} and \eqref{eta-j-Q4},
respectively.
\end{subequations}

\section{Conclusions}

We have shown that infinitely many conservation laws of ABS lattice
equations can be derived from their Lax pairs. We generalized the
approach  used in  \cite{zc}.
From a discrete (two by two) Lax pair it is easy to write out a
formal conservation law. We found a generic  discrete Riccati
equation \eqref{Ri-ger} that is shared by H1, H2, H3, Q1, Q2, Q3 and A1 equation.
This generic  Riccati equation is derived from their Lax pairs. It
provides a series-form  solution for $\theta$, and with the help of
polynomials $\{h_j(\bft)\}$ defined in \eqref{htj}, the infinitely
many conservations laws can be expressed both algebraically and
explicitly. We also want to emphasise that the value of $\beta$ that
we choose is important for getting the solvable generic discrete
Riccati equation,
while in Gardner method $\beta$ is cancelled in the ratio form $\b u=\phi_1/\phi_2$.
Besides, we also note that if we conduct the same procedure starting from $(q,\widehat{~}\,)$ part of Lax pairs,
we will get same conservation laws due to the symmetric property \eqref{sym}.
A2 and Q4 equation seem to be special and so far
we do not know whether their Riccati equations fall in the same
generic  form \eqref{Ri-ger}. For them we derive their conservation laws by using
the approach used for the Ablowitz-Ladik system\cite{zc}. This is closely
related to Gardner method because for a CAC equation, its Lax pair
is obtained by just taking $\b u=\phi_1/\phi_2$ in its BT. However,
starting from Lax pairs gives naturally the formal (initial)
conservation law.
Compared with \cite{Gardner-1}, our formal conservation laws and
the initial conservation laws in \cite{Gardner-1} are same for H1, H2, H3, Q2 and Q3 equation,
while for A1, A2, Q1 and Q4 equation, they are different.
Our approach   can apply to other
multidimentionally consistent systems such as NQC
equation\cite{NQC-1983}, discrete Boussinesq type
equations\cite{Hietarinta-Bous,ZZN-Bous} and so on.

\vskip 10pt
\subsection*{Acknowledgments}
The author (Zhang) are very grateful to Prof. C.W. Cao for his enlightening and enthusiastic
discussion on discrete Lax pairs.
This project is supported by the NSF of China (No. 11071157), SRF of
the DPHE of China (No. 20113108110002) and Shanghai Leading Academic
Discipline Project (No. J50101).


\end{document}